\documentclass[12pt]{iopart}
\expandafter\let\csname equation*\endcsname\relax
\expandafter\let\csname endequation*\endcsname\relax

\usepackage{
    amsfonts,
    amsmath,
    amssymb,
    amsthm,
    bm,
    braket,
    dsfont,
    graphicx,
    latexsym,
    mathpazo,
    mathrsfs, 
    subcaption 
}

\usepackage{hyperref}
\usepackage{amsthm}
\usepackage[autostyle]{csquotes}
\usepackage{bm}

\theoremstyle{definition}

\usepackage{amssymb,amscd}
\usepackage{todonotes}
\usepackage[compatibility=false]{caption}
\usepackage{graphicx}
\usepackage{braket}
\usepackage{dsfont}
\usepackage{color}
\usepackage{tikz}
\usetikzlibrary{arrows,positioning}
\usetikzlibrary{shapes.geometric}
\usetikzlibrary{patterns}
\usepackage{mathrsfs}
\usepackage{algorithm}
\usepackage{algpseudocode}
\algrenewcommand\algorithmicrequire{\textbf{Input:}}
\algrenewcommand\algorithmicensure{\textbf{Output:}}
\usepackage{mathtools}
\usepackage{tcolorbox}
\usepackage[utf8]{inputenc}
\usepackage{enumerate}
\usepackage[english]{babel}
\newtheorem{proposal}{Proposal}
\usepackage{lipsum} 

\algrenewcommand{\Return}{\State\algorithmicreturn~}
\newtheorem{prop}{Proposition}

\algdef{SE}[DOWHILE]{Do}{doWhile}{\algorithmicdo}[1]{\algorithmicwhile\ #1}%
\algnewcommand{\LineComment}[1]{\State \(\triangleright\) #1}

\makeatletter
\renewcommand*\env@matrix[1][*\c@MaxMatrixCols c]{%
	\hskip -\arraycolsep
	\let\@ifnextchar\new@ifnextchar
	\array{#1}}
\makeatother
\renewcommand{\algorithmicrequire}{\textbf{Input:}}
 \renewcommand{\algorithmicensure}{\textbf{Output:}}
 \algnewcommand{\Initialize}[1]{%
  \State \textbf{Initialize:}
  \Statex \hspace*{\algorithmicindent}\parbox[t]{.8\linewidth}{\raggedright #1}
}
 
 \usepackage{xparse}

\NewDocumentCommand{\INTERVALINNARDS}{ m m }{
    #1 {,} #2
}
\NewDocumentCommand{\interval}{ s m >{\SplitArgument{1}{,}}m m o }{
    \IfBooleanTF{#1}{
        \left#2 \INTERVALINNARDS #3 \right#4
    }{
        \IfValueTF{#5}{
            #5{#2} \INTERVALINNARDS #3 #5{#4}
        }{
            #2 \INTERVALINNARDS #3 #4
        }
    }
}
\begin{document}

\title{Randomized benchmarking
    for qudit Clifford gates}
\author{%
Mahnaz Jafarzadeh$^{1,3,6}$,
Ya-Dong Wu$^{2,3,5,6}$,
Yuval R. Sanders$^{4}$
and Barry C. Sanders$^{3,5}$}
\address{$^1$%
Physics Department, Faculty of Sciences, Urmia University, Post Box 165, Urmia, Iran}
\address{$^2$%
QICI Quantum Information and Computation Initiative, Department of Computer Science, The University of Hong Kong, Pok Fu Lam Road, Hong Kong}
\address{$^3$%
Institute for Quantum Science and Technology, University of Calgary, Alberta T2N~1N4,
Canada}
\address{$^4$%
    Department of Physics and Astronomy
        and Australian Research Council Centre of Excellence
            in Engineered Quantum Systems,\\
    Macquarie University, Sydney, New South Wales 2109, Australia
}
\address{$^5$%
Shanghai Branch, National Laboratory for Physical Sciences at Microscale, University of Science and Technology of China, Shanghai
201315, People's Republic of China}
\address{$^6$%
These two authors contributed equally}
\ead{m.jafarzadeh@urmia.ac.ir}
\ead{sandersb@ucalgary.ca}
\date{\today}

\begin{abstract}
We introduce unitary-gate randomized benchmarking (URB) for qudit gates by extending single- and multi-qubit URB to single- and multi-qudit gates. Specifically, we develop a qudit URB procedure that exploits unitary 2-designs. Furthermore, we show that our URB procedure is not simply extracted from the multi-qubit case by equating qudit URB to URB of the symmetric multi-qubit subspace. Our qudit URB is elucidated by using pseudocode, which facilitates incorporating into benchmarking applications.
\end{abstract}

\section{Introduction}
\label{sec:Introduction}
Quantum computing and quantum communication typically focus on quantum information encoded and processed with quantum bit (qubit) strings,
but replacing qubits by higher-dimensional qudit strings~\cite{GKP01,BDS02} can be advantageous~\cite{EFK+18} for quantum simulation~\cite{NAB+09}, quantum algorithms~\cite{LBA+09,TV16,BRS17}, quantum error correction~\cite{DGP13,MSB+16,MLZ+18}, universal optics-based quantum computation~\cite{NCS18}, quantum communication~\cite{CDB+19,LZE+19} and fault-tolerant quantum computation~\cite{CAB12,CAM14}. 
Qudit quantum-information process could reduce space requirements and exploit natural properties such as orbital angular momentum for photons~\cite{BEW+17}, superconductors~\cite{NAB+09} and neutral atoms~\cite{Hec17}. 
Specifically, quantum computing on higher-dimensional systems can be more efficient than on qubits~\cite{BRS17,CAB12,GEZ+19,GP13}.
Moreover, there exist genuine entangled states on higher-dimensional systems that cannot be simulated by 
the tensor product of pairwise entangled qubit states~\cite{KRB+18}.
Ultimate success of quantum computing, both qubit- and qudit-based, depends on being scalable,
which,
in turn,
requires components meeting fault-tolerance conditions~\cite{Sho96}.

Unitary-gate randomized benchmarking (URB)
is the preferred technique to characterize unitary-gate performance due to its efficiency~\cite{MGE11,MGE12},
which is robust against state-preparation-and-measurement (SPAM) errors
and exponentially superior to the alternative of quantum process tomography (QPT)~\cite{MRL08,JLQ+18}.
URB estimates average fidelity between real and ideal implementation of all~$24$ Clifford gates~in the single-qubit Clifford group, $\mathcal{C}_2$,
which normalizes the Pauli group
\begin{equation}
    \bm\Xi_2:=\langle\text{i}
        \mathds1;X_d;Z_d\rangle
\end{equation}
for angular brackets denoting the generating set~\cite{NC10}.
Average unitary-gate fidelity is obtained
by estimating survival decay rate vs gate-sequence length~\cite{MGE12}.
RB is well developed for qubits but untouched for qudits;
here we introduce qudit URB.

We develop qudit URB by combining
single- and multi-qubit URB theory~\cite{KLR+08,MGE11,MGE12,JS14,ATB16,PRY+17,HFGW18,PCR+18}
with qudit
(Hilbert-space dimension~$d$ with $d=2$ for the qubit)
theory including the generalized Pauli group 
$\mathcal{P}_d$~\cite{GKP01,BDS02},
the qudit Clifford group~$\mathcal{C}_d^n$~(\ref{eq:Cdn})~\cite{Got99,App05,Gro06}
for~$n$ qudits
and its connection to the unitary 2-design (U2D)~\cite{GAE07,DCEL09}.
U2Ds are valuable as they enable efficient sampling of a random unitary matrix.
Averaging over a U2D uniform distribution is identical to averaging over the unitary group over the uniform (Haar) measure.

Na\"{i}vely,
qudit URB could be regarded as trivially arising from multi-qubit URB~\cite{MGE12},
which we show is not so,
thereby justifying utilizing U2D properties of the qudit Clifford group.
Our qudit URB scheme reduces to qubit URB as a special case.

Our paper is organized as follows. In \S\ref{sec:Background}, we review background knowledge on generalized Pauli group, qudit Clifford group and U2D for qudits.
Our approach is detailed in \S\ref{sec:Approach}.
We present our results in \S\ref{sec:Results}.
Finally,
\S\ref{sec:Discussion} and \S\ref{sec:Conclusion} provide our discussion and conclusion, repectively.

\section{Background}
\label{sec:Background}
In this section, we provide the required background to address the qudit URB.
As discussing the Clifford group without first an overview of Pauli group is not complete,
we begin the section by explaining generalized Pauli group.
Then we proceed with a discussion of qudit Clifford group.
Finally,
we explain U2D and discuss that qudit Clifford group forms a U2D, which is the most important mathematical concept for qudit URB. 
\subsection{Generalized Pauli group}
\label{Generalized Pauli group}
In this subsection, we describe generalization of Pauli group for qudits.
First, we begin with the mathematical concept of qudits and introducing generalized Pauli operations.
Then we explain generalized Pauli group.
Finally, we describe this group for $n$ qudits.

Mathematically, a qudit is a vector
in $d$-dimensional Hilbert space~$\mathscr{H}_d\cong\mathbb{C}^d$
spanned by the orthonormal computational basis~$\{\ket{s}; s\in\mathbb{Z}_d\}$ where 
$\mathbb{Z}_d:=\{0,1,\dots,d-1\}$.
Qudit unitary transformations are represented by
unitary matrices $\{U_d\}\in U(d)$,
with~$U(d)$ the $d$-dimensional unitary Lie group.
These unitary transformations include generalized Pauli operations~$\bm\Xi_d$, namely,
\begin{equation}
    X_d\ket{s}=\ket{s\oplus1},\,
    Z_d\ket{s}=\omega^s\ket{s},\,
    \omega:=\exp\left(2\pi\text{i}/d\right)
\end{equation}
being defined by their actions on computational basis states
with $\oplus$
denoting addition modulo~$d$~\cite{GKP01,BDS02}. Both $X_d$ and $Z_d$ have order d, obeying $X^d=Z^d=\mathds{1}_d$.
The usual Pauli operators arise for $d=2$.

The generalized Pauli operators $X_d$ and $Z_d$ satisfy the commutation relation 
\begin{equation}
    Z_dX_d=\omega X_dZ_d.
\end{equation}
These operators generate the generalized Pauli group
\begin{equation}
    \mathcal{P}_d:=\langle \tilde{\omega} \mathds{1}_d, X_d,Z_d\rangle,
\end{equation}
for $\tilde{\omega}=\omega$ if $d$ is odd
and $\tilde{\omega}=\omega^{1/2}$ if $d$ is even. If $d$ is odd, the order of $Z_dX_d$ is $d$,
whereas, if~$d$ is even,
then the order of~$Z_dX_d$
is~2$d$, which contributes additional roots of unity.
Therefore, $\tilde{\omega}$ is defined differently for~$d$ even vs odd. The difference between the generalized and usual Pauli operators is that the qudit operators for $d>2$ are only unitary, and not Hermitian.
Hence, any eigenvalue of a qudit Pauli operator can be estimated only via quantum phase estimation algorithm~\cite{NC10}.

The $n$-qudit Hilbert space and the associated state (density operator) space are denoted $\mathscr H_d^{\otimes n}$ and $\mathcal D\left(\mathscr H_d^{\otimes n}\right)$, respectively.
For $n$ qudits,
the generalized Pauli group is
$\mathcal{P}_d^n:=\mathcal{P}_d^{\otimes n}$.
This generalized Pauli group further generalizes to
$\widetilde{\mathcal{P}}_d^n
    :=\mathcal{P}_d^n/\langle\tilde{\omega}\mathds{1}_d\rangle$
without phases for $\langle\tilde{\omega}\mathds{1}_d\rangle$, which indicaties the group generated by $\tilde{\omega}\mathds{1}_d$~\cite{Web16}. 
\subsection{Qudit Clifford group}
\label{subsec:qd Clifford group}
In this subsection, we explain $n$-qudit Clifford group. First we begin the section by describing the Clifford group for qudits. Then we explain the simple case of $d=3$.

As the normalizer of the generalized Pauli group,
the qudit Clifford group comprises all unitary operators that map $\mathcal{P}_d^n$ to itself under conjugation.
Hence, the $n$-qudit Clifford group~(\ref{eq:Cdn})
is~\cite{Web16}
\begin{align}
\label{eq:Cdn}
   \mathcal{C}_d^n
    :=&\{c\in U(d^n) ; c \widetilde{\mathcal{P}}_d^nc^{\dagger}\in \mathcal{P}_d^n \}/ \{\text{e}^{\text{i}\theta}\mathds{1}_d^n; ~\theta\in \mathbb{R}\nonumber\}\\
   =&\langle \text{CZ}_d, \text{F}_d, \text{P}_d, Z_d \rangle,
\end{align}
i.e.,
generated by
controlled-$Z$
\begin{equation}
    \mathrm{CZ}|ss'\rangle:=\omega^{ss'}|ss'\rangle,
\end{equation}
quantum Fourier transform
\begin{equation}
    \mathrm{F}|s\rangle:=\frac{1}{\sqrt{d}}\displaystyle\sum_{s'\in \mathbb{Z}_d}\omega^{ss'}|s'\rangle,
\end{equation}
phase gate
\begin{equation}
    \mathrm{P}|s\rangle:=\omega^{\frac{s(s+\varrho_d)}{2}}|s\rangle,\; \varrho_d= \begin{cases}
1, &\text{if $d$ is odd,}\\
0, &\text{otherwise,}
\end{cases}
\end{equation}
and Pauli-Z gates~\cite{Far14,HDD05,Pro19}.
Any $n$-qudit Clifford gate can be decomposed into multiplicative and tensor products of these gates.
The cardinality of the single-qudit Clifford group has been explicitly calculated~\cite{Tol18}. 
For $d\in\mathbb{P}$ (prime),
the number of distinct Clifford gates (up to global phase) for the single-qudit case is $d^3\left(d^2-1\right)$~\cite{App05}.

For $d=3$, $\mathcal{C}_3$ comprises 216 Clifford gates~\cite{GRT19},
which are generated by single-qutrit Fourier transform~$F_3$ and phase gate~$S_3$.
Specifically, any single-qutrit Clifford operation can be obtained by a product of three elements in
$\mathscr L$,~$\mathscr M$ and~$\mathscr N$~\cite{GRT19},
where~$\mathscr L$ is the subgroup of~$\mathcal C_3$ generated by~$P_3$ and~$X_3$
and
\begin{equation}
    \mathscr M=\{\mathds 1, F_3^2\},\,
    \mathscr N=\{\mathds 1, F_3, S_3F_3, S_3^2 F_3\}
\end{equation}
for~$F_3^2$ and~$S_3^2$
are the squares of~$F_3$ and~$S_3$, respectively.

 Similar to qubit Clifford circuits, quantum circuits with only prime-dimensional qudit Clifford gates
can be classically efficiently simulated~\cite{ME12}. 
To achieve universal quantum computing,
at least one non-Clifford gate must be added to the set of Clifford gates. 
An example of single-qudit non-Clifford gates is the generalized T-gate~\cite{Pro19}
\begin{equation}
\label{eq:T}
    T\ket{s}
        :=\omega^{s^3/d^2}\ket{s}.
\end{equation}
For qubits, T-gates can be benchmarked by dihedral-group benchmarking~\cite{CWE15}, instead of benchmarking U2D,
which is our focus.
Whether dihedral benchmarking can be generalized to the qudit case is an open problem.

\subsection{Unitary 2-design for qudits}
\label{subsec:2-design}
In this subsection, we explain U2D, which is an essential concept for a scalable and efficient URB.
Then we discuss U2D for the multi-qudit Clifford group. 

A U2D comprises unitary matrices 
$\{U_{j}\}_{j=1}^K$
satisfying~\cite{DCEL09}
\begin{equation}\label{eq:2-design}
\frac{1}{K}\displaystyle\sum_{j=1}^K U^{\dagger}_{j}\mathcal E(U_{j}\rho U^{\dagger}_{j})U_{j}=\int_{U(d)}\text{d}U\; U^{\dagger}\mathcal E(U\rho U^{\dagger})U,
\end{equation}
for any quantum channel~$\mathcal E$~\cite{wil11}
and any state~$\rho$ 
with~$\text{d}U$ denoting the unitarily invariant Haar measure~\cite{BR86}
on the Lie group $U(d)$.
Eq.~(\ref{eq:2-design}) implies that twirling any quantum channel over U2D is equivalent to twirling over a Haar-measure unitary group. 
The multi-qubit Clifford group $\mathcal{C}_2^n$ for~$n$ qubits forms a U2D~\cite{DCEL09,DDB02}.

Though not explicitly stated, the multi-qudit Clifford group $\mathcal{C}_d^n$ 
evidently forms a U2D based on Webb's analysis of Pauli-mixing Clifford ensembles~\cite{Web16}.
A set $\mathcal{S} \subseteq \mathcal{C}_d^n$ is called a \emph{Pauli-mixing Clifford ensemble}
(Def.~3 in~\cite{Web16})
if every pair of non-identity Pauli operators are related,
up to a phase,
by a Clifford conjugation
and the number of Clifford operators for each conjugation is constant.
Thus,
a Clifford operator chosen uniformly at random from $\mathcal{S}$ maps 
every non-identity Pauli operator to every other non-identity Pauli operator with equal probability.
Webb's Lemma~3, together with Lemma~2, says that $\mathcal{C}_d^n$ is both Pauli-invariant
(for each $c\in\mathcal S$,
$c \Xi$, up to a phase, is in~$\mathcal S$
as well for all $\Xi\in\mathcal P$)
and Pauli-mixing, and
Lemma~1 says that, if a Pauli-invariant Clifford ensemble is Pauli-mixing,
then this ensemble is a U2D.
Hence, $\mathcal{C}_d^n$ is a U2D.
\section{Approach}
\label{sec:Approach}
In this section, we present our approach to standard URB protocol.
We use the U2D property for the multi-qudit Clifford group to prove twirling over this group is depolarizing.
Then, we discuss averaged sequence fidelity obtained from depolarizng parameter. 
The section ends with explaining the fitting model for URB.
\subsection{Twirling over qudit Clifford group}
In this subsection we show how exploiting U2D property of $n$-qudit Clifford group result in depolarizing channel. Twirling a quantum channel with respect to a group of unitary operations is the basic approach utilized by URB. 
Twirling~$\mathcal E$ with average fidelity
\begin{equation}
    \bar{F}_{\mathcal E}=\int_{\mathscr H_d} \text{d} \phi \bra{\phi}\mathcal E\left(\ket{\phi}\bra{\phi}\right)\ket{\phi}
\end{equation}
over Haar-random unitary operations
yields a depolarizing channel~\cite{HHH99,Nie02}
\begin{equation}
    \int_{U(d)}\text{d}U\;U^{\dagger}\mathcal E(U\rho U^{\dagger})U=\mathcal E_{\text{dep}}(\rho)
\label{twirled channel},
\end{equation}
with the same average fidelity as for~$\mathcal E$,
where
\begin{equation} \label{relationBetweenpandF}
    \mathcal E_{\text{dep}}(\rho):= p\rho+(1-p)\frac{\mathds1}{d^n},\,
    p=\frac{d^n \bar{F}_{\mathcal E}-1}{d^n-1}.
\end{equation}

\begin{prop}\label{prop:twirling}
Twirling a channel over an $n$-qudit Clifford group yields a depolarizing channel.
\end{prop}
\begin{proof}
Combining Eq.~(\ref{twirled channel})
with the fact
that the $n$-qudit Clifford group forms a U2D
yields
\begin{equation}
\label{twirlCliffordDepolarizing}
\frac{1}{|\mathcal C^n_d|}\sum_{l=1}^{|\mathcal C^n_d|} \mathscr{C}_l^{-1} \circ \mathcal E\circ \mathscr{C}_l=\mathcal E_{\text{dep}}.    
\end{equation}
Averaging over the $n$-qudit Clifford group is identical to averaging over the unitary group with respect to the uniform Haar measure
so Eq.~(\ref{twirlCliffordDepolarizing})
follows.
\end{proof}
\subsection{Averaged sequence fidelity} 
In this subsection, first, we set some notation that will be used throughout.
Then we discuss how to generate a random sequence of Clifford gates, and obtain averaged sequence fidelity. 

A concatenation,
or composition,
of $n$-qudit Clifford gates in a sequence of length $m$ is denoted
$\bigcirc_{j=1}^{m} \mathscr{C}_{i_j}$
for $\bigcirc$ indicating concatenation. 
We use $i_j$ as labels for the $j^\text{th}$ element in the $i^\text{th}$ sequence, where $i_j\in\left[\left|\mathcal{C}_d^n\right|\right]$,
with $[k]:=\{1,\dots,k\}$,
and
the gate is denoted~$\mathscr{C}_{i_j}$.
All Clifford gates experience the same noise, 
which is represented by a noisy channel~$\Lambda$ following an ideal Clifford gate.

 Averaging over random realizations of a sequence of Clifford gates
\begin{equation}
    S_{\bm{i}_m}=\bigcirc_{j=1}^m
    \Lambda \circ \mathscr{C}_{i_j}, {\bm{i}_m}:=\left(i_1,i_2,\dots,i_m\right)
\end{equation}
with ${\bm{i}_m}$ denoting an $m$-tuple and $\mathscr{C}_{i_m}=\left(\bigcirc_{j=1}^{m-1} \mathscr{C}_{i_{j}}\right)^{-1}$, is equivalent to concatenating
$m-1$ twirled channels~\cite{MGE11,MGE12} 
\begin{equation}
      \Lambda_{\textrm{T}}:=\frac{1}{|\mathcal C^n_d|}\sum_{l=1}^{|\mathcal C^n_d|} \mathscr{C}_l^{-1} \circ \Lambda\circ \mathscr{C}_l,
\end{equation}
followed by~$\Lambda$,
i.e., $\Lambda\circ \Lambda_{\textrm{T}}^{\circ m-1}$.
Using Proposition~\ref{prop:twirling}, $\Lambda_{\textrm{T}}^{\circ m}$ can be rewritten as an $m$-fold composition of a depolarizing channel with itself multiple times,
namely,
\begin{equation}\label{mfolddepchannel}
    \Lambda_{\textrm{T},p}^{\circ m-1}(\rho)=p^{m-1}\rho+(1-p^{m-1})\mathds{1}/d^n.
\end{equation}
Hence, for any input state~$\ket{\psi}\in \mathscr{H}_d^{\otimes n}$, \begin{equation}\label{eq:IdealSurvivalProbability}
    \tr\left[\ket{\psi}\bra{\psi}\Lambda\circ \Lambda_{\textrm{T},p}^{\circ {m-1}}(\ket{\psi}\bra{\psi})\right]
\end{equation}
is channel fidelity averaged over random realizations of the sequence.
\subsection{Fitting model}
In this subsection, we present the fitting function for URB, by which we model the behaviour of averaged sequence fidelity.
In practice,
with quantum noise, Eq.~(\ref{eq:IdealSurvivalProbability}) is replaced by
\begin{equation}
\label{eq:rearrangeSurvivalProbability}
    F_\text{seq}(m):=\tr\left[E_\psi\Lambda\circ \Lambda_{\textrm{T},p}^{\circ m-1}(\rho_\psi)\right]
\end{equation}
for $E_\psi$ and $\rho_\psi$ the positive-operator valued measure (POVM)~\cite{NC10} element and quantum state
including SPAM errors, respectively.
Plugging Eq.~(\ref{mfolddepchannel}) into Eq.~(\ref{eq:rearrangeSurvivalProbability}) yields
\begin{equation}
\label{fittingModel}
    F_\text{seq}(m)=A_0 p^{m-1}+B_0,
\end{equation}
which absorbs SPAM errors,
for 
\begin{equation*}
A_0:=\tr\left[E_\psi\Lambda\left(\rho_\psi-\frac{\mathds{1}}{d^n}\right)\right],\,
B_0:=\frac{\tr\left[E_\psi\Lambda(\mathds{1})\right]}{d^n}
\end{equation*}
being the coefficients.
\subsection{Summary of approach}
\label{subsec:summary}
In this section, exploiting U2D property of $n$-qudit Clifford group we presented twirling a channel over this group yields a depolarizing channel. Then, we obtained averaged sequence fidelity from depolarizing rate. Furthermore, we introduced fitting model for our URB.
This section relies on the assumption that quantum noise is gate-independent.
As gate-dependent noise decays in the same form
as gate-independent noise plus a perturbation~\cite{W18},
our approach could be naturally extended to the case of gate-dependent noise.

\section{Results}
\label{sec:Results}
In this section we present our main results. The first result of this paper is providing pseudocode for URB procedure. We proceed to explain the URB procedure for qudit Clifford gates. Then, giving a counterexample, we explain why multi-qubit URB does not readily yield qudit URB. 
\subsection{Randomized benchmarking as an algorithm}
\label{subsec:algorithm}
In this subsection, we provide pseudocode for our multi-qudit URB procedure,
which is immensely useful to ensure that the procedure flows logically and does not leave out any key steps.
Our algorithm
is designed to estimate average gate fidelity~$\bar{F}_\Lambda$
over a Haar-random set of input states.

Our pseudocode uses the following data types,
expressed conventionally as all capitals.
UNSIGNED INTEGER refers to a positive integer in~$\mathbb{Z}^+$,
and REAL, COMPLEX, BINARY, DARY and INTERVAL refer to real~$\mathbb R$, complex~$\mathbb C$,
binary~$\{0,1\}$,
d-ary~$\{0,1,\dots,d-1\}$
and the unit interval~$[0,1]\subset\mathbb{R}$,
respectively.
Besides classical data types, we introduce quantum data types as well~\cite{knil96}.
QDARY
refers to a qudit of dimension~$d$.
Each of these data types can be an array
with data type followed by brackets~$[~]$;
a sequence of two brackets $[~][~]$
denotes a two-dimensional array,
which is readily generalized to higher-dimensional arrays by adding more brackets.

For pseudocode variables, we use \textsc{camelCase},
so qudit number~$n$
is denoted by \textsc{numQud}
and of type UNSIGNED INTEGER,
and Hilbert-space dimension~$d$
is \textsc{HilbDim}
and also of type UNSIGNED INTEGER.
We use \textsc{numSeq} to denote the maximum number of different gate sequences of a fixed length and of type UNSIGNED INTEGER.
QDARY[~] indicates a multi-qudit state and
QDOP is an operation that maps a QDARY[\textsc{HilbDim}]-typed variable to another QDARY[\textsc{HilbDim}]-typed variable. 

In describing our algorithm using pseudocode,
we employ functions from an ideal library explained here.
We use \textsc{rand}(\textsc{maxInt}) to generate a uniformly random integer in [\textsc{maxInt}].
\textsc{prod} maps two Clifford-gate indices to the index corresponding
to the product of these two referenced Clifford gates.
\textsc{inv} maps one Clifford-gate index to the index corresponding to the  
inverse of the Clifford gate.       
We use data type \textsc{prep}
for preparing a qudit pure state according to a classical description of the state, and
 \textsc{projMeas} denotes qudit-state measurement that yields~$1$ if the qudit state is projected onto a certain pure state and otherwise yields~$0$.
For statistical processes,
we employ \textsc{avg},
which calculates the average of all entries in an array,
and \textsc{fit},
which is a least-squares regression algorithm.
Our randomized benchmarking algorithm comprises input, output and procedure,
which we now describe in plain English.
\subsubsection{Input}
We begin by explaining the input.
The input can be separated into two components,
those that are necessary to specify the benchmarking task and
those that are necessary to specify the benchmarking procedure.
To specify the benchmarking task, the user must specify 
the number of qudits $n \in \mathbb{Z}^+$,
the dimension of the Hilbert space $d \in \mathbb{Z}^+$,
and the cardinality $\left|\mathcal C^n_d\right|\in\mathbb{Z}^+$
of the $n$-qudit Clifford group.
Cardinality growing quickly for qudits implies significant experimental challenges.
Circuit depth for testing just the qutrit Clifford group, comprising 216 elements,
would be challenging.
We need a length-$\left|\mathcal C^n_d\right|$ array
of labels for Clifford operators
in order to be able to refer to them individually.

For the benchmarking procedure,
our algorithm caters to an experienced client
who is able to guess good parameters for the number~$k\in\mathbb{Z}^+$ of random Clifford 
sequences,
the number $l\in\mathbb{Z}^+$ of repetitions of each 
Clifford sequence
and the maximum length $m\in\mathbb{Z}^+$ of a random 
Clifford sequence to be executed.
This requirement that the client be able to select good parameters is typical for all qubit-based quantum benchmarking~\cite{JS14}.
Furthermore,
the client is expected to know that the noise model is specified
by an unknown CPTP map, and therefore knows that
the twirled noise model is entirely specified by a single unknown depolarising parameter.
The client aims to estimate this unknown parameter to within a target confidence,
which is not an input to the algorithm.
\subsubsection{Output}
The output of the algorithm is an estimate~$\hat{r}$ of the average gate infidelity, often called average error rate,
\begin{equation}
    r:=1- \bar{F}_\Lambda\in[0,1].
\end{equation}
This estimate is a URB figure of merit 
that characterizes average performance of qudit Clifford gates.
\subsubsection{Procedure}
Now we explain the URB procedure
for qudit Clifford gates.
We initialize the $n$-qudit state~$\ket{\psi}$
as the pure state $\ket0^{\otimes n}$.
Then we generate $k$ random sequences of $n$-qudit Clifford gates
$\mathscr{C}_{i_j}$ 
each of length~$j$, 
where $2\le j\le m$, as~$k$ samples of a random sequence. The first~$j-1$ gates in each sequence are uniformly randomly chosen from~$\mathcal C_d^n$,
and the final gate is determined by the first~$j-1$ gates according to
\begin{equation}
    \mathscr{C}_{i_j}=\left(\bigcirc_{j^\prime=1}^{j-1} \mathscr{C}_{i_{j^\prime}}\right)^{-1}.
\end{equation}
As Clifford gates form a group, this final gate is also an element of the group.
    
We apply each of the $k$ sequences of qudit gates to the initial state.
Then
we apply measurements corresponding to the POVM~$\left\{\Ket{\psi}\bra{\psi}, \mathds{1}-\Ket{\psi}\Bra{\psi}\right\}$ 
on the output state. 
If the measurement outcome corresponds to~$\ket{\psi}\bra{\psi}$,
we assign a value of one,
otherwise, a value of zero.
By averaging over~$k$ different sequences and~$l$ copies of each sequence, we obtain an estimate $\hat{F}(j)$ of the averaged sequence fidelity
\begin{equation}
\label{survivalprobabilityqubit}
        F_\text{seq}(j)=\tr(E_{\psi}\mathcal S_{j}(\rho_\psi)),\,
    \mathcal S_{j}=\frac{1}{k} \sum_{\bm{i}_{\bm{j}-1}\in \left[\left|\mathcal C_d^n\right|\right]^{\otimes j-1} } S_{\bm{i}_j}.
\end{equation}
Now we repeat the above procedure for different values of~$j$, which increases from two to the maximum length~$m$
in succession.
Finally, we fit the estimates $\hat{F}(j)$ to Eq.~(\ref{fittingModel})
with~$p$ the decay parameter
and~$\hat{p}$ its estimate.

Per Eq.~(\ref{relationBetweenpandF}), we see
that fidelity decay parameter $p$ is related to~$r$
    via
\begin{equation}\label{eq:averageErrorRate}
    r=(1-p)\left(1-\frac{1}{d^n}\right).
\end{equation}
Therefore, by estimating $p$ from URB of Clifford gates, we obtain the output average infidelity (algorithm~\ref{alg:quditRB}).
\begin{algorithm}[H]
\caption{Randomized Benchmarking} \label{alg:quditRB}
\begin{algorithmic}[1]
\Require{%
\Statex UNSIGNED INTEGER \textsc{numQud}
\Comment{\# qudits}
\Statex UNSIGNED INTEGER \textsc{HilbDim}
\Comment{Hilbert-space dimension}
\Statex UNSIGNED INTEGER \textsc{cardCliff} \Comment{Cardinality of \textsc{numQud}-qudit Clifford group }
\Statex UNSIGNED INTEGER \textsc{gateIndex}[\textsc{cardCliff}]
\Comment{Array of multi-qudit gates' labels}
\Statex UNSIGNED INTEGER  \textsc{maxLeng} \Comment{Maximal sequence length}
\Statex UNSIGNED INTEGER \textsc{numSeq}
\Comment{\# sequences for each length}
\Statex UNSIGNED INTEGER \textsc{numCop}
\Comment{\# copies for each gate sequence}
}
\Ensure{\Statex  INTERVAL \textsc{avInfid}
\Comment{Average infidelity}
    }
 \Procedure{randBench}{\textsc{numQud},
 \textsc{HilbDim}, \textsc{CliffGate}, \textsc{maxLeng}, \textsc{numSeq}, \textsc{numCop}}

 \State COMPLEX \textsc{initState}$\left[\textsc{HilbDim}\wedge\textsc{numQud}\right] \gets [0,0,\dots,0]$;
  \Comment{The initial state is the \textsc{numQud}-qudit pure state}
\State REAL \textsc{estRegr}[3];
\Comment{Base, slope and intercept for exponential fiting}
\State REAL \textsc{seqFid}[\textsc{maxLeng}];
 \Comment{Fidelity averaged over gate sequences}
 \State UNSIGNED INTEGER \textsc{currIndex}; 
    \State QDARY \textsc{state}[\textsc{numQud}];
 \For{$k=2:\textsc{maxLeng}$} 
  \State BINARY \textsc{outcome}$\left[\textsc{numCop}\right]$;
 \Comment{Array of measurement outcomes}
    \State  REAL \textsc{survProb}[\textsc{numSeq}];
     \Comment{Array of survival probabilities of the initial state}
 \For{$i=1:\textsc{numSeq}$}
 \State $\textsc{currIndex}\gets 1$; 
 \Comment{Overwrite $\textsc{currIndex}$ with $1$, the index refered to indentity channel}
    \For{$j=1:k-1$}
     \State $\textsc{gateIndex}[j]
   \gets \textsc{rand}\left(\textsc{cardCliff}\right)$;
    \Comment{Generate a random integer from one to \textsc{cardCliff}.}
     \State $\textsc{currIndex}\gets \textsc{prod}(\textsc{currIndex},\textsc{gateIndex})$;
    \EndFor
         \State $\textsc{gateIndex}\left[k\right]\gets \textsc{inv}(\textsc{currIndex})$; 
     \For{$l=1:\textsc{numCop}$}
    \State $\textsc{state} \gets \textsc{prep}(\textsc{initState}$);
 \For{$j=1:k$}
 \State QDOP \textsc{CliffGate}[\textsc{cardCliff}];
\Comment{Array of multi-qudit gates}
    \State $\textsc{state} \gets \textsc{CliffGate}\left[\textsc{gateIndex}[j]\right]*\textsc{state}$;
    \Comment{$j$th Clifford gate maps old state to new state}
    \EndFor
     \algstore{myalg}
\end{algorithmic}
\end{algorithm}
\begin{algorithm}                     
\begin{algorithmic}           
\algrestore{myalg}
    \If{$\textsc{projMeas}(\textsc{state}, \textsc{initState}) =1$ }  
    \State $\textsc{outcome}[l]\gets 1$;\Comment{Assign~$1$ if each single-qudit state is projected onto~\textsc{initState}}
    \Else
    \State $\textsc{outcome}[l]\gets 0$;\Comment{Otherwise, assign~$0$}
    \EndIf
\EndFor
\State $\textsc{survProb}[i]\gets \textsc{avg}(\textsc{outcome})$;
\EndFor
\State $\textsc{seqFid}[k]\gets \textsc{avg}(\textsc{survProb})$; 
\EndFor
\State \textsc{estRegr}
$\gets \textsc{fit}\left( \textsc{seqFid}\right)$; \Comment{Least-squares regression algorithm with exponential fitting model $y=\textsc{estRegr}[2]*\textsc{estRegr}[1]\wedge x+\textsc{estRegr}[3]$.}
\State \Return $\textsc{avInfid}\gets (1-\textsc{estRegr}[1])
*\left(1-1/\textsc{HilbDim}\wedge\textsc{numQud}\right)$.
  \EndProcedure
\end{algorithmic}
\end{algorithm}

\subsection{Multi-qubit randomized benchmarking versus qudit randomized benchmarking}
Now we explain why previous work on multi-qubit URB does not readily yield qudit URB.
One might expect that solving multi-qubit URB would yield qudit URB trivially.
Such an approach would exploit Schur-Weyl duality~\cite{GW09}.
Schur-Weyl duality,
applied to the symmetry group~$S_n$ and the unitary group $U(d)$,
which have commuting actions on the $n$-fold tensor product of $d$-dimensional Hilbert spaces,
$\mathscr H_d^{\otimes n}$,
states that,
under the joint action of~$S_n$ and~$U(d)$,
the tensor product space decomposes into a direct sum of tensor products of irreducible modules.
The question is whether we can use that to construct Clifford operators for qudits. 
We show that this enticing notion is fallacious by falsifying the following proposal.
\begin{proposal}\label{conj:tensor}
A tensor product of any~$n$ single-qubit Clifford operators is a direct sum of a Clifford operator for the $(n+1)$-dimensional symmetric space with any operators for the remaining partially and antisymmetric spaces.
\emph{(FALSE)}
\end{proposal}
\noindent This proposal is enticing because we could simply use existing multi-qubit benchmarking work~\cite{MGE11} instead of producing a new result.

Mathematically,
this proposal can be expressed as follows.
Let~$\left\{C^{(2)}_\imath\right\}_{\imath=1}^n$
be a sequence of~$n$ single-qubit operators,
and let~$C^{(n+1)}$
be any Clifford operator on an $(n+1)$-dimensional Hilbert space~$\mathscr{H}_d$ for $d=n+1$.
The conjecture is then that,
for all $\left\{C^{(2)}_\imath\right\}_{\imath=1}^n$,
a Clifford operator~$C^{(n+1)}$ exists
such that
\begin{equation}
\bigotimes_\imath C^{(2)}_\imath=C^{(n+1)}\oplus \Theta^{(2^n-n-1)},    
\end{equation}
for any $(2^n-n-1)$-dimensional unitary operator $\Theta^{(2^n-n-1)}$.
We now demonstrate that this proposal is false by giving a counterexample.

\begin{proof}[\textsl{Counterexample}]
This proposal is falsified with a counterexample,
specifically for $C^{(2)}_\imath=H$
for $\imath\in\{1,2\}$.
The two-qubit Hadamard gate is
\begin{equation}
    H\otimes H
    =R\oplus(-1),\;
R:=\frac{1}{2}
\begin{pmatrix}
 1&\sqrt2&1\\
 \sqrt2&0&-\sqrt2\\
1&-\sqrt2&1
 \end{pmatrix},
\end{equation}
which is block-diagonal on both the symmetric and anti-symmetric subspaces.
However, the block part $R$ on the three-dimensional symmetric subspace is not a qutrit Clifford gate, as $RX_3R\notin \mathcal P_3$.
\end{proof}
\noindent Therefore, the proposal is falsified:
a tensor product of qubit Clifford gates
cannot in general be written as a direct sum involving Clifford gates,
and qudit Clifford gates are not directly
obtained from multi-qubit Clifford gates over the symmetric subspace. This falsification implies a significant difference between quantum computing on qudits vs on multiple qubits,
even for the same total dimension. 

\section{Discussion}
\label{sec:Discussion}
We have explained how we can characterize the average performance of the qudit Clifford gates directly through performing qudit URB. The procedure is similar to the qubit URB, just the unitary operators are chosen from qudit Clifford group. We have designed URB for qudit Clifford gates, by synthesizing the U2D property of them with qubit URB. We also devise a pseudocode, which provides the instructions on how to run randomized benchmarking algorithm on a quantum computer.

The U2D property of the qudit Clifford group indicates that twirling a noisy channel over this group yields a depolarizing channel per Proposition~\ref{prop:twirling}. Hence, analogous to qubit URB, we can relate the depolarizing parameter, estimated from URB procedure, to average error rate of qudit Clifford gates. 
On the other hand, the natural question that arises is,
given that multi-qubit URB has already been studied, whether qudit URB could be determined from multi-qubit case by considering the symmetric subspace of multi-qubits. We have explained that this symmetrization in Conjecture~\ref{conj:tensor} fails.
 \section{Conclusion}
\label{sec:Conclusion}
Our results extend previous URB results to higher dimensional qudits for estimating average error rate for gate independent errors, and pave the way for experimental characterization of qudit Clifford gates. Recent development on photonic qudit-based quantum computing~\cite{NCS18} provides a good test bed for our qudit URB.
We suggest that quantum optics will provide a good test by exploiting different
photonic degrees of freedom, for example orbital angular-momentum~\cite{BEW+17}, frequency~\cite{LLP+18,IJA+19}, and time~\cite{IJA+19}.

\ack
BCS acknowledges financial support from the Natural Sciences and Engineering Research Council of Canada and from the National Natural Science Foundation of China (NSFC) Grant No.~11675164.
YDW acknowledges support from the Hong Kong Research Grant Council through Grant No.~17300918.
YRS is funded by Australian Research Council Discovery Project DP190102633. 

\bibliography{bibfile}
\end{document}